\documentclass[conference]{llncs}
\usepackage{graphicx}
\usepackage[utf8]{inputenc}
\usepackage{paralist}
\usepackage{tikz}
\usepackage{rotating}
\usepackage{setspace}
\usepackage{mathtools}
\usepackage{bussproofs}
\usepackage{latexsym}
\usepackage{amsmath, amsfonts, amssymb}

\begin{document}

\title{Resolution in Linguistic First Order Logic based on Linear Symmetrical Hedge Algebra}
\author{Thi-Minh-Tam Nguyen\inst{1} 
			\and Viet-Trung Vu\inst{2} \and\\
			 The-Vinh Doan\inst{2}
			\and  Duc-Khanh Tran\inst{3}}
			
\institute{Vinh University\\
			\email{nmtam@vinhuni.edu.vn}
			\and 	Hanoi University of Science and Technology\\
		\email{
				trungvv91@gmail.com, 
				doanthevinh1991@gmail.com}
			\and Vietnamese German University\\
				\email{khanh.td@vgu.edu.vn} } 

\maketitle

\begin{abstract}

This paper focuses on resolution in linguistic first order logic with truth value taken from linear symmetrical hedge algebra. We build the basic components of linguistic first order logic, including syntax and semantics. We present a resolution principle for our logic to resolve on two clauses having contradictory linguistic truth values. Since linguistic information is uncertain, inference in our linguistic logic is approximate. Therefore, we introduce the concept of reliability in order to capture the natural approximation of the resolution inference rule. 

\end{abstract}

\section{Introduction}

Automated reasoning theory based on resolution rule of Robinson \cite{Robinson65} has been research extensively in order to find efficient proof systems \cite{Chang1997,Klement}. However, it is difficult to design intelligent systems based on traditional logic while most of the information we have about the real world is uncertain. Along with the development of fuzzy logic, non-classical logics became formal tools in computer science and artificial intelligence. Since then, resolution based on non-classical logic (especially multi-valued logic and fuzzy logic) has drawn the attention of many researchers.

In 1965, Zadeh introduced fuzzy set theory known as an extension of set theory and applied widely in fuzzy logic \cite{Zadeh65}. Many researchers have presented works about the fuzzy resolution in fuzzy logic \cite{Ebrahim01,Lee1972,Mondal,Shen89,Vojtas01,WeigertTL93}. In 1990, Ho and Wechler proposed an approach to linguistic logic based on the structure of natural language \cite{Wechler}. The authors introduced a new algebraic structure, called hedge algebra, to model linguistic truth value domain, which applied directly to semantics value in inference. There also have been many works about inference on linguistic truth value domain based on extended structures of hedge algebra such as linear hedge algebra, monotony linear hedge algebra \cite{LeVH2009,Phuong_Khang,LAPhuong_GMP}. Researchers also presented truth functions of new unary connectives (hedges) from the set of truth values to handle fuzzy truth values in a natural way
 \cite{Esteva2013,Hajek2001,Vychodil2006}.
Recently, we have presented the resolution procedure in linguistic propositional logic with truth value domain taken from linear symmetrical hedge algebra \cite{LSHA}. We have constructed a linguistic logic system, in which each sentence in terms of \textit{``It is very true that Mary studies very well'}' is presented by $P^\mathsf{VeryTrue}$, where P is \textit{``Mary studies very well''}. Two clauses having converse linguistic truth values, such as $P^\mathsf{VeryTrue}$ and $P^\mathsf{MoreFalse}$, are resolved by a resolution rule. However, we cannot intervene in the structure of a proposition. For example with the knowledge base: \textit{``It is true that if a student studies hard then he will get the good marks''} and \textit{``It is very true that Peter studies hard''}, we cannot infer to find the truth value of the sentence \textit{``Peter will get the good marks''}. Linguistic first order logic overcomes this drawback of linguistic propositional logic. Furthermore, knowledge in the linguistic form maybe compared in some contexts, such as when we tell about the value of linguistic variable \textit{Truth}, we have $\mathsf{LessTrue<VeryTrue}$ or $\mathsf{MoreFalse<LessFalse}$. Therefore, linear symmetrical hedge algebra is an appropriate to model linguistic truth value domain.

As a continuation of our research works on resolution in linguistic propositional logic systems \cite{LSHA,RHA}, we study resolution in linguistic first order logic. We construct the syntax and semantics of linguistic first order logic with truth value domain taken from linear symmetrical hedge algebra. We also propose a resolution rule and a resolution procedure for our linguistic logic. Due to the uncertainty of linguistic information, each logical clause would be associated with a certain confidence value, called reliability. Therefore, inference in our logic is approximate. We shall build an inference procedure based on resolution rule with a reliability $\alpha $ which ensures that the reliabilities of conclusions are less than or equal to reliabilities of premises. 

The paper is structured as follows:  section \ref{sec:HA} introduces basic notions of linear symmetrical hedge algebras and logical connectives. Section \ref{sec:syntax-semantics} describes the syntax and semantics of our linguistic first order logic with truth value domain based on linear symmetrical hedge algebra. Section \ref{sec:resolution} proposes a resolution rule and a resolution procedure. Section \ref{sec:conclusion} concludes and draws possible future work.

\section{Linear Symmetrical Hedge Algebra}
\label{sec:HA}

We present here an appropriate mathematical structure of a linguistic domain called hedge algebra which we use to model linguistic truth domain for our linguistic logic. In this algebraic approach, values of the linguistic variable \textit{Truth} such as $\{True, MoreTrue, VeryPossibleTrue, PossibleFalse, LessFalse\}$, and so on are generated from a set of generators (primary terms) G =$\{False,True\}$ using hedges from a set $H = \{Very,More, Possible,Less,...\} $ as unary operations. There exists a natural ordering among these values, with $a \leq b$ meaning that a indicates $a$ degree of truth less than or equal to $b$, where $a < b$ iff $a \leq b$ and $a \neq b$. . For example, $True < VeryTrue$ and $False < LessFalse$. The relation $ \leq$ is called the semantically ordering relation on the term domain, denoted by $X$.

In general, $X$ is defined by an abstract algebra called hedge algebra $HA = (X, G, H, >)$ where $G$ is the set of generators and $H$ is the set of hedges. The set of values $X$ generated from $G$ and $H$ is defined as $X = \{ \delta c | c \in G, \delta \in H \}$. $\geq$ is a partial order  on X such that $a \geq b$ if $a>b$ or $a=b$ ($a,b \in X)$.

Each hedge $h \in H$ either strengthens or weakens the meaning of a term $x \in X$, this means $hx$ and $x$ are always comparable. For example, $VeryTrue > True$ but $PossibleTrue < True$. Therefore, the set H can be decomposed into two subsets: one subset $H^+$ consists of hedges which strengthen the primary term $True$ and the other, denoted by $H^-$, consists of hedges that weaken the term $True$.

Each hedge has a strengthening or weakening degree w.r.t. linguistic terms and so the sets $H^+$ and $H^-$ maybe ordered; and they then become a poset (partially ordered set). The ordering relationship
between two hedges $h$ and $k$ will induce relationship between $hx$ and $kx$ for every $x$ in $X$. For example, as $Less < More$, we have $LessPossibleTrue < MorePossibleTrue$.

Hedges are modifiers  which change the meaning of a term $x$ only a little. Therefore, if $h$ is a hedge, the meaning of term $hx$ must inherit the one of $x$. For every term x; if the meaning of hx and kx can be expressed by the ordering relationship $hx < kx$; then $\delta hx < \delta' kx$; for any strings of hedges $\delta $ and $\delta'$. For example from $PossibleFalse < LessF alse$ it follows that $V eryPossibleFalse < VeryPossibleLessFalse$.

Let $h, k$ be two hedges in the set of hedges $H$. Then $k$ is said to be \textit{positive (negative)} w.r.t. $h$ if for every $x \in X$, $hx \geq x$ implies $khx \geq hx (khx \leq hx)$ or, conversely, $hx \leq x$ implies $khx \leq hx ( khx \geq hx )$.
$h$ and $k$ are \textit{converse} if $\forall x \in X, hx \leq x$ iff $kx \geq x$, i.e. they are in the different subset. $h$ and $k$ are \textit{compatible} if $\forall x \in X, x \leq hx$ iff $x \leq kx$, i.e. they are in the same subset. $h$ modifies terms stronger or equal than $k$, denoted by $h \geq k$, if $\forall x \in X, (hx \geq kx \geq x)$ or $(hx \geq kx \geq x)$.

Given a term $u$ in $X$, the expression $h_n \ldots h_1u$ is called a representation of $x$ w.r.t. $u$ if $x=h_n \ldots h_1u$, and it is called a canonical representation of $x$ w.r.t. $u$ if $h_nh_{n-1} \ldots h_1u \neq h_{n-1} \ldots h_1u$. The notation $x_{u|j}$ denotes the suffix of length j of a representation of $x$ w.r.t. $u$. The following propositioin shows how to compare any two terms in $X$.
\begin{proposition} \cite{Wechler}
Let $x=h_nh_{n-1}\ldots h_1u$, $y=k_mk_{m-1}\ldots k_1 u$ be two canonical presentations of $x$ and $y$ w.r.t. $u \in X$, respectively. Then, there exists the largest $j \leq min(m,n)+1$ such that $\forall i < j, h_i=k_i$, and
\begin{enumerate}[i.]
\item $x = y$ iff  $m = n$ and $h_{j}x_{u|j} = k_{j}x_{u|j}$ for every $j \leq n$;
\item $x < y$ iff  $h_{j}x_{u|j} < k_{j}x_{u|j}$;
\item $x$ and $y$ are incomparable iff $h_{j}x_{u|j}$ and $k_{j}x_{u|j}$
\end{enumerate}
\end{proposition}

The set of primary terms $G$ usually consists of two comparable ones, denoted by $c^- < c^+$. For the variable \textit{Truth}, we have $c^+ = True > c^- = False$. Such HAs are called \textit{symmetric} ones. For symmetric HAs, the set of hedges $H$ is decomposed into two disjoint subsets $H^+$ and $H^-$ defined as $H^+=\{h \in H| hc^+ > c^+\}$ and $H^-=\{h \in H| hc^+ < c^+\}$. Two hedges in each of the sets $H^+$ and $H^-$ maybe comparable or incomparable. Thus, $H^+$ and $H^-$ become posets.

\begin{definition} \cite{Programming}
A symmetric HA $AX=(X, G = \{c^-, c^+\}, H, \leq)$ is called a \textit{linear symmetric HA} (\textit{lin-HA}, for short) if the set of hedges $H$ is devided into two subsets $H^+$ and $H^-$, where $H^+=\{h \in H| hc^+ > c^+\}$, $H^-=\{h \in H| hc^+ < c^+\}$, and $H^+$ and $H^-$ are linearly ordered.
\end{definition}
Let $x$ be an element of the hedge algebra $AX$ and the canonical representation of $x$ is $x = h_{n}...h_{1}a$ where $a \in \{c^+, c^-\}$. The contradictory element of $x$ is an element $\overline{x}$ such that $\overline{x} = h_{n}...h_{1}a'$ where $a' \in \{c^+, c^-\}$ and $a' \neq a$.. In  \textit{lin-HA}, every element $x \in X$ has an unique contradictory element in $X$.

HAs are extended by augmenting two hedges $\Phi$ and $\Sigma$ defined as $\Phi(x)=infimum(H(X))$ and $\Sigma(x)=supremum(H(x))$, for all $x \in X$ \cite{Ho1992}. An HA is said to be free if $\forall x \in X$ and $\forall h \in H, hx \neq x$. It is shown that, for a free \textit{lin-HA} of the variable \textit{Truth} with $H \neq \emptyset, \Phi(c^+)=\Sigma(c^-)$, $\Sigma(c^+)=\top$ (AbsolutelyTrue), and $\Phi(c^-)=\bot$ (AbsolutelyFalse). Let us put $W = \Phi(c^+)=\Sigma(c^-)$ (called the middle truth value), we have $\bot<c^-<W<c^+<\top$.

\begin{definition}
\label{defi:truth_domain}
A linguistic truth domain $\overline{X}$ taken from a \textit{lin-HA} $AX=(X,\\ \{c^-, c^+\},H, \leq)$ is defined as $\overline{X} = X \cup \{\bot, W, \top \}$, where $\bot, W, \top$ are the least, the neutral, and the greatest elements of $\overline{X}$, respectively.
\end{definition}

\begin{proposition}
\cite{Ho1992} For any \textit{lin-HA} $AX = (X,G,H,\leq)$, the linguistic truth domain $\overline{X}$ is linearly ordered.
\end{proposition}

In many-valued logic, sets of connectives called {\L}ukasiewicz, G\"{o}del, and product logic ones are often used. Each of the sets has a pair of residual t-norm and implicator.  However, we cannot use the product logic connectives when our truth values are linguistic.
We showed that the logical connectives based on
G\"{o}del's t-norm and t-conorm operators are more suitable for our linguistic logic than those based on {\L}ukasiewicz's \cite{LSHA} . Therefore, in this paper we define logical connectives using G\"{o}del's t-norm and t-conorm operators \cite{Graded,WeigertTL93}. 

Let $K=\{n|n\in \mathbb{N} , n\leq N_0\}$. 
A pair of $(T,S)$ in G\"{o}del's logic is defined as follows:
\begin{itemize}
\item $T_{G}(m,n)=\min(m,n)$.
\item $S_{G}(m,n)=\max(m,n)$.
\end{itemize}
It is easy to prove that $T_G, S_G$ are commutative, associate, monotonous.

Given a \textit{lin-HA} AX, since all the values in AX are linearly ordered, truth functions for conjunctions and disjunctions are G\"{o}del's t-norms and t-conorms, respectively.

\begin{definition}
\label{logical_connections}
Let $S$ be a linguistic truth domain, which is a \textit{lin-HA} $AX=(X, G,
H,\leq)$, where $G = \{\top, \mathsf{True}, \mathsf{W},
\mathsf{False}, \bot \}$. The logical connectives $\wedge$ (respectively $\vee$) over the set $X$ are defined to be G\"{o}del's t-norm (respectively t-conorm), and furthermore to satisfy the following:
$\neg \alpha = \overline \alpha$, and 
$\alpha \rightarrow \beta = (\neg \alpha)\vee \beta$,
where $\alpha,\beta \in X$.
\end{definition}

\begin{proposition}
\label{pros:property_connection}
Let $S$ be a linguistic truth domain, which is a \textit{lin-HA} $AX=(X,\{\top, \mathsf{True}, \mathsf{W}, \mathsf{False}, \bot \}, H,\leq)$; $\alpha,\beta, \gamma \in X$, we have:
\begin{itemize}
	\item Double negation:
	$\neg (\neg \alpha) = \alpha$

	\item Commutative:
	 $\alpha \wedge \beta = \beta \wedge \alpha$, 
	$\alpha \vee \beta = \beta \vee \alpha$

	\item Associative:
	$(\alpha \wedge \beta)\wedge \gamma = \alpha\wedge(\beta\wedge\gamma)$,
	$(\alpha \vee \beta)\vee \gamma = \alpha\vee(\beta\vee\gamma)$

	\item Distributive:
 $\alpha\wedge(\beta\vee\gamma)=(\alpha\wedge\beta)\vee(\alpha\wedge\gamma)$, 
 $\alpha\vee(\beta\wedge\gamma)=(\alpha\vee\beta)\wedge(\alpha\vee\gamma)$
	
		\end{itemize}
		\end{proposition}

\section{Linguistic First Order Logic based on Linear Symmetrical Hedge Algebra}
\label{sec:syntax-semantics}

In this section we define the syntax and semantics of our linguistic first-order logic.
\subsection{Syntax}
\begin{definition}
The alphabet of a linguistic first-order language consists of the following sets of symbols:
\begin{itemize}
\item constant symbols: a set of symbols $a, b, c, \ldots$, each of 0-ary;
\item logical constant symbols: $\mathsf{MoreTrue}, \mathsf{VeryFalse}, \bot, \top, ...$;
\item variable: $x, y, z, \ldots$;
\item predicate symbols: a set of symbols $P, Q, R, \ldots$, each associated with a positive integer n, arity. A predicate with arity n is called n-ary;
\item function symbols: a set of symbols $f, g, h, \ldots$, each associated with a positive integer n, arity. A function with arity n is called n-ary;
\item logical connectives: $\lor, \land,  \neg, \rightarrow, \leftrightarrow$; 
\item quantifies: universal quantification $\forall$, existentional quantification $\exists$;
\item auxiliary symbols: $ \Box, (, ), \ldots$.
\end{itemize}
\end{definition}

\begin{definition}
A term is defined recursively as follows:
\begin{itemize}
\item either every constant or every variable symbol is a term,
\item if $t_{1}, \ldots, t_{n}$ are terms and $f$ is a n-ary function symbol, $f(t_{1}, \ldots, t_{n})$ is a term (functional term).
\end{itemize}
\end{definition}

\begin{definition}
An atom is either a zero-ary predicate symbol or a n-ary predicate symbol $P(t_{1}, \ldots, t_{n})$, where $t_{1}, \ldots, t_{n}$ are terms.
\end{definition}

\begin{definition}
Let A be an atom and $\alpha$ be a logical constant. Then $A^\alpha$ is called a literal to represent $A$ is $\alpha$.
\end{definition}

\begin{definition}
Formulae are defined recursively as follows:
\begin{itemize}
	\item a literal is a formula,
	\item if $F, G$ are formulae, then $F \lor G$, $F \land G$, $F\rightarrow G, F\leftrightarrow G, \neg F$  are formulae, and
	\item if $F$ is a formula and x is a free variable in $F$, then $(\forall x)F$ and $(\exists x)F$ are formulae.
\end{itemize}
\end{definition}
The notions of free variable, bound variable, substitution, unifier, most general unifier, ground formula, closed formula, etc. are similar to those of classical logic.

\begin{definition}
A clause is a finite disjunction of literals represented by $ L_{1} \lor L_{2} \lor ... \lor L_{n}$, where $L_{i} (i=1,2,...,n)$ is a literal. An empty clause is denoted by $\Box$.
\end{definition}
A formula is in conjunctive normal form (CNF) if it is a conjunction of clauses. It is well known that transforming a formula in first order logic into a CNF formula preserves satisfiability \cite{Chang1997}. In Section \ref{sec:resolution} we shall be working with a resolution procedure which processes CNF formulae, or equivalently clause sets.

\subsection{Semantics}


\begin{definition}
\label{interpretation}
An interpretation for the linguistic first order logic is a pair I=$<$D,A$>$ where $D$ is a non empty set called domain of $I$, and $A$ is a function that maps:
	\begin{itemize}
		\item every constant symbol $c$ into an element $c^A \in D$;
		
		\item every n-ary function symbol f into a function $f^A: D^n \rightarrow X$;
		
		\item every logical constant symbol $l$ into an element $l^A \in X$;
		\item every n-ary predicate symbol P into an n-ary relation $P^A: D^n \rightarrow X$, where X is the truth value domain taken from \textit{lin-HA};
		\item every variable x into a term.
	\end{itemize}
\end{definition} 
Given an interpretation \textit{I=$<$D,A$>$} for the linguistic first order logic, the truth value of a symbol $S$ in the alphabet of the logic is denoted by $I(S)$.

\begin{definition} Given an interpretation I=$<$D,A$>$, we define: 
\begin{itemize} 
\item Value of a term: $I(t) = t^A$,
$I(f(t_1, \ldots, t_n)) = f(I(t_1), \ldots, I(t_n))$.
\item Truth value of an atom: $I(P(t_1, \ldots, t_n)) = P(I(t_1), \ldots, I(t_n))$.
\item Truth value of a logical constant: $I(c) = c^A$.
\item Let $P$ be an atom such that $I(P) = \alpha_1$. Truth value of a literal $P^{\alpha_2}$:
	
	\[
I(P ^{\alpha_2})= 
	\begin{cases}
		\alpha_1 \land \alpha_2  \text{ if }  \alpha_1, \alpha_2 > \mathsf{W}, \\
		\neg (\alpha_1 \lor \alpha_2) \text{ if } \alpha_1, \alpha_2 \leq \mathsf{W}, \\ 
		(\neg \alpha_1) \lor \alpha_2, \text{ if } \alpha_1 > \mathsf{W}, \alpha_2 \leq W, \\
		\alpha_1 \lor (\neg \alpha_2), \text{ if } \alpha_1 \leq \mathsf{W}, \alpha_2 > \mathsf{W}.
	\end{cases}	
\]
\item Let F and G be formulae. Truth value of a formula:
\begin{multicols}{2}
	\begin{itemize}
	\item $I(\neg F) = \neg I(F)$
	\item $I(F \land G) = I(F) \land I(G)$
	\item $I(F \lor G) = I(F) \lor I(G)$
	\item $I(F \rightarrow G) = I(F) \rightarrow I(G)$
	\item $I(F \leftrightarrow G) = I(F) \leftrightarrow I(G)$
	\item $I((\forall x)F) = min_{\forall d \in D}\{I(F)\}$
	\item $I((\exists x)F) = max_{\exists d \in D}\{I(F)\}$
	\end{itemize}	
\end{multicols}
\end{itemize}
\end{definition}

\begin{definition}
Let I=$<$D,A$>$ be an interpretation and $F$ be a formula. Then
\begin{itemize}
\item $F$ is true iff $I(F) \geq W$. $F$ is satisfiable iff there exists an interpretation $I$ such that $F$ is true in $I$ and we say that $I$ is a model of $F$ (write $I \models F$) or $I$ satisfies $F$.
\item $F$ is false iff $I(F)<W$ and we say that $I$ falsifies $F$. $F$ is unsatisfiable iff there exists no interpretation that satisfies $F$.
\item $F$ is valid iff every interpretation of $F$ satisfies $F$.
\item A formula $G$ is a logical consequence of formulas $\{F_1, F_2, \ldots, F_n\}$ iff for every interpretation $I$, if $I \models F_1 \land F_2 \land \ldots \land F_n$ we have that $I \models G$.
\end{itemize}
\end{definition}

\begin{definition}
\label{defi:equivalent}
Two formulae F and G are logically equivalent iff $F\models G$ and $G \models F$ and we write $F \equiv G$.
\end{definition}

It is infeasible to consider all possible interpretations over all domains in order to prove the unsatisfiability of a clause set $S$. Instead, we could fix on one special domain such that $S$ is unsatisfiable iff $S$ is false under all the interpretations over this domain. Such a domain, which is called the Herbrand universe of $S$, defined as follows.

\label{label:HUniver}
Let $H_0$ be the set of all constants appearing in $S$. If no constant appears in $S$, then $H_0$ is to consist of a single constant, say $H_0 = \{a\}$. For $i=0,1,2,\ldots$, let $H_{i+1}$ be the union of $H_i$ and the set of all terms of the form $f^n(t_1,\ldots,t_n)$ for all $n$-place functions $f^n$ occurring in $S$, where $t_j$, $j=1,\ldots,n$, are members of the set $H_i$. Then each $H_i$ is called the i-level constant set of $S$ and $H_\infty$ is called the Herbrand universe (or H-universe) of $S$, denoted by $H(S)$.

\label{defi:Hbase}
The set of ground atoms of the form $P^n(t_1,\ldots,t_n)$ for all n-ary predicates $P^n$ occuriring in $S$, where $t_1,\ldots,t_n$ are elements of the H-universe of $S$, is called the atom set, or Herbrand base (H-base, for short) of $S$, denoted by $A(S)$.

A ground instance of a clause $C$ of a clause set $S$ is a clause obtained by replacing variables in $C$ by members of H-universe of $S$.

We now consider interpretations over the H-universe. In the following we define a special over the H-universe of $S$, called the H-interpretation of $S$.

\begin{definition}
\label{defi:Hinter}
Let $S$ be a clause set, $H$ be the H-universe of S, and I=$<$D,A$>$ be an interpretation of $S$.  $\mathcal{I}$ is an H-interpretation of $S$ if the following holds:
\begin{itemize}
\item $D=H$,
\item Let c be a constant symbol, $c^A = c$,
\item Let f be a n-ary function symbol, $f^A$ maps $(h_1,\ldots,h_n) \in H^n$ to $f(h_1,\ldots,h_n) \in H$
\item Let $A=\{A_1, \ldots, A_n, \ldots\}$ be the H-base (or atom set) of $S$, H-interpretation $\mathcal{I} = \{m_1, \ldots, m_n, \ldots\}$, where $m_j = A_j$ or $m_j = \neg A_j$.
\end{itemize}
\end{definition}
Given $I=<D,A>$ interpretation over $D$, an H-interpretation $\mathcal{I}=<H,\mathcal{A}>$ corresponding to $I$ is an H-interpretation that satisfies the following condition:

Let $h_1, \ldots, h_n$ be elements of $H$ and let $m:H \rightarrow D$ be a mapping from $H$ to $D$ then $P^\mathcal{A}(h_1,\ldots,h_n) = P^A(m(h_1),\ldots,m(h_n))$

Given an Interpretation $I$, we can always find a corresponding $\mathcal{I}$ H-interpretation.
\begin{lemma}
\label{lemma:Hsatis}
If an interpretation I over some domain D satisfies a clause set S, then any one of the H-interpretations $\mathcal{I}$ corresponding to I also satisfies S.
\end{lemma}
\begin{proof} 
Assume $\mathcal{I}$ falsifies $S$ over domain $D$. Then there must exist at least one clause $C$ in $S$ such that $\mathcal{I}(C)<\mathsf{W}$. 
Let $x_1, \ldots, x_n$ be the variables occurring in $C$. Then there exist $h_1,\ldots,h_n$ in $H(S)$ such that $\mathcal{I}(C')<\mathsf{W}$ where $C'$ is ground clause obtained from $C$ by replacing every $x_i$ with $h_i$.
Let every $h_i$ mapped to some $d_i$ in $D$ by $I$. By the definition of H-interpretation of $S$ in Def. \ref{defi:Hinter}, if $C''$ is the ground clause obtained from $C$ by replacing every $x_i$ with $d_i$ then $I(C'') < \mathsf{W}$. This means that $I$ falsifies $S$ which is impossible.
\end{proof}
\begin{theorem}
\label{theo:Hunsatis}
A clause set S is unsatisfiable iff S is false under all the \\H-interpretations of S.
\end{theorem}
\begin{proof}
($\Rightarrow$) Obviously, by definition S is unsatisfiable iff S is false under all the interpretations over any domain.

($\Leftarrow$) Assume that $S$ is false under all the H-interpretations of S. Suppose S is satisfiable. Then there is an interpretation $I$ over some domain $D$ such that $I(S) \geq \mathsf{W}$. Let $\mathcal{I}$ be an H-interpretation corresponding to $I$. According to Lemm. \ref{lemma:Hsatis}, $\mathcal{I}(S) \geq \mathsf{W}$. This contradicts the assumption that $S$ is false under all the H-interpretations of $S$. Therefore, $S$ must be unsatisfiable.
\end{proof}
Let $S$ be a clause set and $A(S)$ be the H-base of $S$. A \textit{semantic tree} for $S$ is a  complete binary tree constructed as follows:
\begin{itemize}
\item For each node $N_i$ at the $i^th$ level corresponds to an element $A_i$ of $A(S)$, that is, the left edge of $N_i$ is labeled $A_i < \mathsf{W}$, the right edge of $N_i$ is labeled $A_i \geq \mathsf{W}$.

\item Conversely, each element of $A(S)$ corresponds to exactly one level in the tree, this means if $A_i \in A(S)$ appears at level $i$ then it must not be at any other levels.
\end{itemize}

%
%
%
%

Let $T$ be a semantic tree of a clause set $S$ and $N$ be a node of $T$. We denote $\mathcal{I}(N)$ to be the union of all the sets labeled to the edges of branch of $T$ down to $N$. If there exists an H-interpretation $\mathcal{I}$ in $T$ which contains $\mathcal{I}(N)$, such that $\mathcal{I}(N)$ falsifies some ground instance of $S$, then $S$ is said to be failed at the node $N$. A node $N$ is called a \textit{failure node} of $S$ iff $S$ falsifies at $N$ and $\mathcal{I}(N')$ does not falsify any ground instance of a clause in $S$ for every ancestor node $N'$ of $N$. $N$ is called an \textit{inference node} if all the immediate descendant nodes of N are failure nodes.
If every branch in $T$ contains a failure node, cutting off its descendants from $T$, we have $T'$ which is called a \textit{closed tree} of $S$. If the number of nodes in $T'$ is finite, $T'$ is called a finite closed semantic tree.

\begin{lemma}
\label{lemma:inferencenode}
There always exists an inference node on finite closed tree.
\end{lemma}
\begin{proof}
Assume that we have a closed tree $CT$. Because $CT$  has finite level, so there exists at least one leaf node $j$ on $CT$  at the highest level. Let $i$ be parent node of $j$.  By definition of closed tree, $i$ cannot be failure node. Therefore, $i$ has another child node, named $k$. If $k$ is a failure node then $i$ is inference node, the lemma is proved. If $k$ is not a failure node then it has two child nodes: $l, m$. Clearly $l, m$ are at higher level than $j$. This contradicts with the assumption that $j$ is at the highest level. Therefore $k$ is a failure node and $i$ is an inference node. The lemma is proved.

\begin{figure}
\label{fig:inferencenode}

\centering
\begin{tikzpicture}[-, main node/.style={circle,fill=white!25,draw,node distance=2.5cm},minimum width=3pt]

  \node[main node] (1) {i};
  \node[main node] (2) [below left of=1] {k};
  \node[main node] (3) [below right of=1] {j};
  \node[main node] (4) [below left of=2] {l};
  \node[main node] (5) [below right of=2] {m};

  \path[every node/.style={font=\sffamily\small}]
  	(1) edge node [left] {} (2)
  		edge node [right] {} (3)
  	(2) edge node [left] {} (4)
  		edge node [right] {} (5)
	;

\end{tikzpicture}
\caption{Proof of inference node}
		\end{figure}
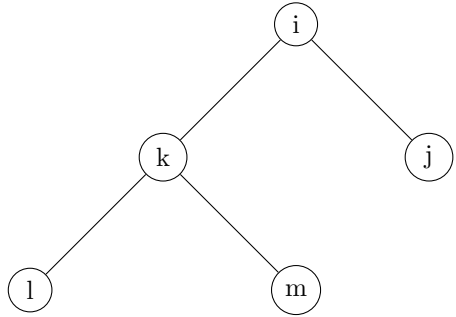
\end{proof}
\begin{lemma}
\label{theo:HerbrandI}
Let $S$ be a clause set. Then $S$ is unsatisfiable iff for every semantic tree of $S$, there exists a finite closed tree.
\end{lemma}
\begin{proof}
$(\Rightarrow)$ Suppose $S$ is unsatisfiable and $T$ is a semantic tree of $S$. For each branch $B$ of $T$, let $\mathcal{I}_B$ be the set of all literals labeled to all edges of the branch $B$ then $\mathcal{I}_B$ is an H-interpretation for $S$. Since $S$ is unsatisfiable, $\mathcal{I}_B$ must falsify a ground instance $C'$ of a clause $C$ in $S$. However, since $C'$ is finite, there must exists a failure node $N_B$ on the branch $B$. Since every branch of $T$ has a failure node, there is a closed semantic tree $T'$ for $S$. Furthermore, since only a finite number of edges are connected to each node of $T'$, the number of nodes in $T'$ must be finite, for otherwise, by K$\ddot{o}$nig Lemma, we could find an infinite branch containing no failure node. Thus, $T'$ is a finite closed tree.

$(\Leftarrow)$ Conversely, if corresponding to every semantic tree $T$ for $S$ there is a finite closed semantic tree, by the definition of closed tree, every branch of $T$ contains a failure node. This means that every interpretation falsifies $S$. Hence $S$ is unsatisfiable.
\end{proof}
In the next section we present the inference based on resolution rule for our linguistic logic. Lemma \ref{lemma:inferencenode} and Lemma \ref{theo:HerbrandI} will be used to prove the soundness and completeness of resolution inference rule.

\section{Resolution}
\label{sec:resolution}

\hspace*{0.3cm} 
In two-valued logic, when we have a set of formulae $\{ A, \neg A \}$ (written as $\{ A^\mathsf{True}, A^\mathsf{False} \}$ in our logic) then the set is said to be contradictory. However in our logic, the degree of contradiction can vary because the truth domain contains more than two elements. Let us consider two sets of formulae $\{ A^\mathsf{VeryTrue}, A^\mathsf{VeryFalse} \}$ and $\{ A^\mathsf{LessTrue}, A^\mathsf{LessFalse} \}$. Then the first set of formulae is ``more contradictory'' than the second one. Consequently, the notion of reliability is introduced to capture the approximation of linguistic inference.
 
\begin{definition}
Let $\alpha$ be an element of $X$ such that $\alpha > \mathsf{W}$ and $C$ be a clause. The clause $C$ with a reliability $\alpha$ is denoted by the pair $(C,\alpha)$.
\end{definition}

The reliability $\alpha$ of a clause set $S=\{ C_1,C_2,\ldots,C_n\}$ is defined as follows:
$\alpha = \alpha_1 \land \alpha_2 \land \ldots \land \alpha_n, $
where $\alpha_i$ is the reliability of $C_i$ $(i=1,2,\ldots,n)$.

\label{defi:variant}
A clause $(C_2,\alpha_2)$ is a variant of a clause $(C_1,\alpha_1)$ if $\alpha_1 \neq \alpha_2$ or $C_2$ is equal to $C_1$ except for possibly different variable name.

\subsection{Fuzzy linguistic resolution}

\label{defi:factor}
The clause $C_2$ is a factor of clause $C_1$ iff $C_2 = C_1 \sigma$, where $\sigma$ is a most general unifier (m.g.u, for short) of some subset $\{L_1,\ldots,L_k\}$ of $C_1$.

\begin{definition}
\label{defi:luathopgiai}
Given two clauses $(C_1,\alpha_1)$ and $(C_2,\alpha_2)$ without common variables, where $C_1 = A^a \lor C'_1$, $C_2 = A^a \lor C'_2$.
Define the linguistic resolution rule as follows:
$$\frac{(A^a \lor C'_1,\alpha_1) \hspace{0.5cm} (B^b \lor C'_2,\alpha_2)}{(C'_1 \gamma \lor C'_2 \gamma,\alpha_3)}$$
where a, b, and $\alpha_3$ satisfy the following conditions:
\begin{displaymath}
\left \{
\begin{array}{l}
a \land b < \mathsf{W} ,\\
a \lor b \geq \mathsf{W} ,\\
\gamma \text{ is an m.g.u of } A \text{ and } B,\\
\alpha_3 = f(\alpha_1, \alpha_2, a, b),
\end{array}
\right.
\end{displaymath}
with $f$ is a function ensuring that $\alpha_3 \leq \alpha_1$, and $\alpha_3 \leq \alpha_2$.
\\$(C'_1 \gamma \lor C'_2 \gamma,\alpha_3)$ is a binary resolvent of $(C_1,\alpha_1) \text{ and } (C_2,\alpha_2)$. The literals $A^a$ and $B^b$ are called literals resolved upon.
\end{definition}

In Def. \ref{defi:luathopgiai}, $\alpha_3$ is defined so as to be smaller or equal to both $\alpha_1$ and $\alpha_2$. In fact, the obtained clause is less reliable than original clauses. The function $f$ is defined as following:
\begin{align}
\label{formula:combine-reliability}
\alpha_3 = f(\alpha_1, \alpha_2, a, b) = \alpha_1 \land \alpha_2 \land (\neg(a \land b)) \land (a\lor b)
\end{align}
Obviously, $\alpha_1,\alpha_2 \ge \mathsf{W}$, and $\alpha_3$ depends on $a, b$. Additionally, $a \land b < \mathsf{W}$ implies $\neg (a \land b) >\mathsf{W}$. Moreover, $(a \lor b) \geq \mathsf{W}$.  Then, by Formula (\ref{formula:combine-reliability}), we have $\alpha_3 \geq \mathsf{W}$.

An inference is sound if its conclusion is a logical consequence of its premises. That is, for any interpretation I, if the truth values of all premises are greater than $\mathsf{W}$, the truth value of the conclusion must be greater than $\mathsf{W}$.

\begin{definition}
\label{defi:resolvent}
A resolvent of clauses $C_1$ and $C_2$ is a binary resolvent of factors of $C_1$ and $C_2$, respectively.
\end{definition}

\begin{definition}
\label{defi:derivable}
Let $S$ be a clause set. A resolution derivation is a sequence of the form $S_0,\ldots,S_i, \ldots$, where
\begin{itemize}
\item $S_0 = S$, and
\item $S_{i+1} = S_i \cup \{(C,\alpha)\}$, where $(C,\alpha)$ is the conclusion of a resolution inference with premises $S_i$ based on resolution rule in Def. \ref{defi:luathopgiai} and $(C,\alpha) \notin S_i$.
\end{itemize}
\end{definition}

\begin{lemma}[Lifting lemma]
\label{lemma:Lifting}
If $C_1'$ and $C_2'$ are instances of $C_1$ and $C_2$, respectively, and if $C'$ is a resolvent of $C_1'$ and $C_2'$, then there is a resolvent $C$ of $C_1$ and $C_2$ such that $C'$ is an instance of $C$.
\end{lemma}
\begin{proof}
Let $C_1=A^a \lor C_1'$ and $C_2=B^b \lor C_2'$.

$C'_1 = {\Gamma'_1}^\alpha \lor {T'_1}^{\beta_1},
 C'_2 = {\Gamma'_2}^\delta \lor {T'_2}^{\beta_2}$ 
($\beta_1 \land \beta_2 < \mathsf{W}, \beta_1 \lor \beta_2 > \mathsf{W}$), 
$\gamma$ is a m.g.u of $T'_1,T'_2$. 
$\sigma$ is an assignment.
 
$C'_1 = C_1\sigma, C'_2 = C_2\sigma$ where $C_1 = {\Gamma_1}^\alpha \lor {T_1}^{\beta_1}, C_2 = {\Gamma_2}^\delta \lor {T_2}^{\beta_2}$. By resolution rule \ref{defi:luathopgiai}, $C' = \gamma o \sigma ({\Gamma'_1}^\alpha \lor {\Gamma'_2}^\delta) = \gamma o \sigma ({\Gamma_1}^\alpha \lor {\Gamma_2}^\delta)$ because of $\Gamma'_1 = \Gamma_1 \sigma, \Gamma'_2 = \Gamma_2 \sigma$. Assume $\omega$ is a m.g.u of $T_1,T_2$ then $\omega$ is more general then $\gamma$, implying $\omega$ is more general $\gamma o \sigma$. Hence, $C' = \gamma o \sigma ({\Gamma_1}^\alpha \lor {\Gamma_2}^\delta)$ is an instance of $C = \omega({\Gamma_1}^\alpha \lor {\Gamma_2}^\delta)$. The lemma is proved.
\end{proof}

We find that resolution derivation $S_0,\ldots,S_i, \ldots$ is infinite because the set of assignments and the set of semantic values are infinite. However, if the original clause set $S$ is unsatisfiable, the sequence $S_i$ always derives an empty clause $\Box$. The soundness and completeness of resolution derivation is shown by the following theorem: 
\begin{theorem}
\label{theo:completeness}
Let $S$ be a clause set, $S_0, \ldots, S_i, \ldots$ be a resolution derivation. $S$ is unsatisfiable iff there exists $S_i$ containing the empty clause $\Box$.
\end{theorem}
\begin{proof}
$(\Rightarrow)$ Suppose $S$ is unsatisfiable. Let $A=\{A_1, A_2,\ldots\}$ be the atom set of $S$. Let $T$ be a semantic tree for $S$. By Theo. \ref{theo:HerbrandI}, $T$ has a finite closed semantic tree $T'$. 

If $T'$ consists of only one root node, then $\Box$ must be in $S$ because no other clauses are falsified at the root of a semantic tree. 
Thus the theorem is true.

Assume $T'$ consists of more than one node, by Lemm. \ref{lemma:inferencenode} $T'$ has at least one inference node. Let $N$ be an inference node in $T'$, and let $N_1$ and $N_2$ be the failure nodes immediately below $N$. 

Since $N_1$ and $N_2$ are failure nodes but $N$ is not a failure node, there must exist two ground instances $C_1'$ and $C_2'$ of clauses $C_1$ and $C_2$ such that $C_1'$ and $C_2'$ are false in $\mathcal{I}(N_1)$ and $\mathcal{I}(N_2)$, respectively, but both $C_1'$ and $C_2'$ are not falsified by $\mathcal{I}(N)$. Therefore, $C_1'$ must contain a literal $A^a$ and $C_2'$ must contain a literal $B^b$ such that $\mathcal{I}(A^a) < \mathsf{W}$ and $\mathcal{I}(B^b) \geq \mathsf{W}$.

Let $C' = (C_1' - A^a) \lor (C_2' - B^b)$. $C'$ must be false in $\mathcal{I}(N)$ because both $(C_1' - A^a)$ and $(C_2' - B^b)$ are false. By the Lifting Lemma we can find a resolvent $C$ of $C_1$ and $C_2$ such that $C'$ is a ground instance of $C$.

Let $T''$ be the closed semantic tree for $(S \cup \{C\})$ obtained from $T'$ by deleting any node or edge that is below the first node where the resolvent $C'$ is falsified. Clearly, the number of nodes in $T''$ is fewer than that in $T'$. Applying the above process on $T''$, we can obtain another resolvent of clauses in $(S \cup \{C\})$. Putting this resolvent into $(S \cup \{C\})$ we can get another smaller closed semantic tree. This process is repeated until the closed semantic tree consists of only the root node. This is possible only when $\Box$ is derived, therefore there is a deduction of $\Box$ from $S$.

$(\Leftarrow)$ Suppose there is a deduction of $\Box$ from $S$. Let $R_1, \ldots, R_k$ be the resolvents in the deduction. Assume $S$ is satisfiable then there exists $\mathcal{I} \models S$. If a model satisfies clauses $C_u$ and $C_v$, it must also satisfy any resolvent of $C_u$ and $C_v$. Therefore $\mathcal{I} \models (C_u \land C_v)$.
Since resolution is an inference rule then if $\mathcal{I} \models (C_u \land C_v)$ then $\mathcal{I} \models R_i$ for all resolvents.
However, one of the resolvents is $\Box$ therefore $S$ must be unsatisfiable. The theorem is proved.
\end{proof}

A \emph{resolution proof} of a clause $C$ from a set of clauses $S$ consists of repeated application of the resolution rule to derive the clause $C$ from the set $S$. If $C$ is the empty clause then the proof is called a \emph{resolution refutation}. We shall represent resolution proofs as \emph{resolution trees}. Each tree node is labeled with a clause. There must be a single node that has no child node, labeled with the conclusion clause, we call it is the root node. All nodes with no parent node are labeled with clauses from the initial set $S$. All other nodes must have two parents and are labeled with a clause $C$
such that $$\frac{C_1 \hspace{0.5cm} C_2}{C}$$ where $C_1, C_2$ are the labels of the two parent nodes. If $\mathsf{RT}$ is a resolution tree
representing the proof of a clause with reliability $(C, \alpha)$,
then we say that $\mathsf{RT}$ has the reliability $\alpha$.

\begin{example}
\label{eg:RT}
Let $AX = (X, G, H, \leq, \neg,\lor,\land , \rightarrow)$ be a \textit{lin-HA} where $G=\{\bot,\mathsf{False},$ $\mathsf{W},\mathsf{True},\top\}$,
 $\bot,\mathsf{W},\top$ are the smallest, neutral, biggest elements, respectively, and
 $\bot<\mathsf{False} < \mathsf{W} < \mathsf{True} < \top$; $H^+ =$ $\{\mathsf{V}$,$\mathsf{M}\}$ and $H^- = \{\mathsf{P, L}\}$ (V=Very, M=More, P=Possible, L=Less); 
Consider the clause set after transforming into CNF as following:
\begin{multicols}{2}
\begin{enumerate}
\item $A(x)^\mathsf{MFalse} \lor B(z)^\mathsf{MFalse} \lor C(x)^\mathsf{PTrue}$
\item $C(y)^\mathsf{MFalse} \lor D(y)^\mathsf{VMTrue}$
\item $C(t)^\mathsf{VVTrue} \lor E(t,f(t))^\mathsf{MFalse}$
\columnbreak
\item $E(a, u)^\mathsf{True}$
\item $A(a)^\mathsf{VTrue}$
\item  $B(a)^\mathsf{LTrue}$
\item $D(a)^\mathsf{MFalse}$
\end{enumerate}
\end{multicols}
where a, b are constant symbols; t, x, y, u, z are variables.
At the beginning, each clause is assigned to the highest reliability $\top$. We have two of resolution proofs as follows:

\begin{prooftree}
\label{eg:RT2}
\AxiomC{($A(x)^\mathsf{MFalse} \lor B(z)^\mathsf{MFalse} \lor C(x)^\mathsf{PTrue},\top$)}
\AxiomC{($A(a)^\mathsf{VTrue}, \top$)}
\kernHyps{.2in}\insertBetweenHyps{\hskip-.0in}
\RightLabel{${\scriptstyle [a/x]}$}
\BinaryInfC{$(B(z)^\mathsf{MFalse} \lor C(a)^\mathsf{PTrue}, \mathsf{MTrue})$}

\AxiomC{($B(a)^\mathsf{LTrue}, \top$)}
\insertBetweenHyps{\hskip-.99in}

\RightLabel{${\scriptstyle [a/z]}$}
\BinaryInfC{$(C(a)^\mathsf{PTrue}, \mathsf{LTrue})$}
\AxiomC{($C(y)^\mathsf{MFalse} \lor D(y)^\mathsf{VMTrue}, \top$)}
\insertBetweenHyps{\hskip-.999in}

\RightLabel{${\scriptstyle [a/y]}$}
\BinaryInfC{$(D(a)^\mathsf{VMTrue}, \mathsf{LTrue})$}
\AxiomC{($D(a)^\mathsf{MFalse}, \top$)}
\insertBetweenHyps{\hskip-.9990in}

\BinaryInfC{$(\Box, \mathsf{LTrue})$}
\end{prooftree}

\begin{prooftree}
\label{ex:RT3}
\AxiomC{$(C(y)^\mathsf{MFalse} \lor D(y)^\mathsf{VMTrue},\top)$}
\AxiomC{$(D(a)^\mathsf{MFalse},\top)$}
\kernHyps{.2in}\insertBetweenHyps{\hskip-.0in}
\RightLabel{${\scriptstyle [a/y]}$}
\BinaryInfC{$(C(a)^\mathsf{MFalse},\mathsf{MTrue})$}

\AxiomC{$(C(t)^\mathsf{VVTrue} \lor E(t,f(t))^\mathsf{MFalse},\top)$}
\insertBetweenHyps{\hskip-.99in}

\RightLabel{${\scriptstyle [a/t]}$}
\BinaryInfC{$(E(a,f(a))^\mathsf{MFalse},\mathsf{MTrue})$}
\AxiomC{($E(a,u)^\mathsf{True}, \top$)}
\insertBetweenHyps{\hskip-.999in}

\RightLabel{${\scriptstyle [f(a)/u]}$}
\BinaryInfC{$(\Box, \mathsf{True})$}
\end{prooftree}
\end{example}

\section{Conclusion}
\label{sec:conclusion}
We have presented syntax and semantics of our linguistic first order logic system. We based on linear symmetrical hedge algebra to model the truth value domain. To capture the approximate of inference in nature language, each clause in our logic is associated with a reliability. We introduced an inference rule with a reliability which ensures that the reliability of the inferred clause is less than or equal to those of the premise clauses. Based on the algebraic structure of linear symmetrical hedge algebra, resolution in linguistic first order logic will contribute to automated reasoning on linguistic information. It would be worth investigating how to extend our result to other hedge algebra structures and to other automated reasoning methods. 

\bibliographystyle{plain} 
\bibliography{References}

\begin{thebibliography}{10}

\bibitem{Chang1997}
Chin-Liang Chang and Richard Char-Tung Lee.
\newblock {\em Symbolic Logic and Mechanical Theorem Proving}.
\newblock Academic Press, Inc., Orlando, FL, USA, 1st edition, 1997.

\bibitem{Ebrahim01}
Rafee Ebrahim.
\newblock Fuzzy logic programming.
\newblock {\em Fuzzy Sets and Systems}, 117(2):215--230, 2001.

\bibitem{Esteva2013}
Francesc Esteva, Lluís Godo, and Carles Noguera.
\newblock A logical approach to fuzzy truth hedges.
\newblock {\em Information Sciences}, 232(0):366 -- 385, 2013.

\bibitem{Ho1992}
Nguyen~Cat Ho and Wolfgang Wechler.
\newblock Extended hedge algebras and their application to fuzzy logic.
\newblock {\em Fuzzy Sets and Systems}, 52(3):259 -- 281, 1992.

\bibitem{Hajek2001}
Petr Hájek.
\newblock On very true.
\newblock {\em Fuzzy Sets and Systems}, 124(3):329 -- 333, 2001.
\newblock Fuzzy Logic.

\bibitem{Klement}
Erich~Peter Klement.
\newblock Some mathematical aspects of fuzzy sets: triangular norms, fuzzy
  logics, and generalized measures.
\newblock {\em Fuzzy Sets Syst.}, 90(2):133--140, September 1997.

\bibitem{LeVH2009}
Van~hung Le, Fei Liu, and Dinh~khang Tran.
\newblock Fuzzy linguistic logic programming and its applications.
\newblock {\em Theory Pract. Log. Program.}, 9(3):309--341, May 2009.

\bibitem{Programming}
Van~Hung Le, Fei Liu, and Dinh~Khang Tran.
\newblock Fuzzy linguistic logic programming and its applications.
\newblock {\em TPLP}, 9(3):309--341, 2009.

\bibitem{Lee1972}
Richard C.~T. Lee.
\newblock Fuzzy logic and the resolution principle.
\newblock {\em J. ACM}, 19(1):109--119, January 1972.

\bibitem{Mondal}
B.~Mondal and S.~Raha.
\newblock Approximate reasoning in fuzzy resolution.
\newblock In {\em Fuzzy Information Processing Society (NAFIPS), 2012 Annual
  Meeting of the North American}, pages 1--6, Aug 2012.

\bibitem{Wechler}
C.H. Nguyen and W.~Wechler.
\newblock {\em Hedge Algebras: An Algebraic Approach in Struture of Sets of
  Linguistic Truth Values}, pages 281--293.
\newblock Fuzzy Sets and Syst. 35, 1990.

\bibitem{LSHA}
Thi-Minh-Tam Nguyen, Viet-Trung Vu, The-Vinh Doan, and Duc-Khanh Tran.
\newblock Resolution in linguistic propositional logic based on linear
  symmetrical hedge algebra.
\newblock In {\em Knowledge and Systems Engineering}, volume 244 of {\em
  Advances in Intelligent Systems and Computing}, pages 327--338, 2014.

\bibitem{Phuong_Khang}
Le~Anh Phuong and Tran~Dinh Khang.
\newblock A deductive method in linguistic reasoning.
\newblock In {\em Uncertainty Reasoning and Knowledge Engineering (URKE), 2012
  2nd International Conference on}, pages 137--140, 2012.

\bibitem{LAPhuong_GMP}
Le~Anh Phuong and Tran~Dinh Khang.
\newblock Linguistic reasoning based on generalized modus ponens with
  linguistic modifiers and hedge moving rules.
\newblock In {\em Fuzzy Theory and it's Applications (iFUZZY), 2012
  International Conference on}, pages 82--86, 2012.

\bibitem{Robinson65}
John~Alan Robinson.
\newblock A machine-oriented logic based on the resolution principle.
\newblock {\em J. ACM}, 12(1):23--41, 1965.

\bibitem{Shen89}
Z.~Shen, L.~Ding, and M.~Mukaidono.
\newblock Fuzzy resolution principle.
\newblock In {\em Multiple-Valued Logic, 1988., Proceedings of the Eighteenth
  International Symposium on}, pages 210--215, 1988.

\bibitem{Graded}
Dana Smutná-Hliněná and Peter Vojtáš.
\newblock Graded many-valued resolution with aggregation.
\newblock {\em Fuzzy Sets and Systems}, 143(1):157 -- 168, 2004.

\bibitem{RHA}
Duc-Khanh Tran, Viet-Trung Vu, The-Vinh Doan, and Minh-Tam Nguyen.
\newblock Fuzzy linguistic propositional logic based on refined hedge algebra.
\newblock In {\em Fuzzy Systems (FUZZ), 2013 IEEE International Conference on},
  pages 1--8, 2013.

\bibitem{Vojtas01}
Peter Vojt{\'a}s.
\newblock Fuzzy logic programming.
\newblock {\em Fuzzy Sets and Systems}, 124(3):361--370, 2001.

\bibitem{Vychodil2006}
Vilém Vychodil.
\newblock Truth-depressing hedges and bl-logic.
\newblock {\em Fuzzy Sets and Systems}, 157(15):2074 -- 2090, 2006.

\bibitem{WeigertTL93}
Thomas~J. Weigert, Jeffrey J.~P. Tsai, and Xuhua Liu.
\newblock Fuzzy operator logic and fuzzy resolution.
\newblock {\em J. Autom. Reasoning}, 10(1):59--78, 1993.

\bibitem{Zadeh65}
Lotfi~A. Zadeh.
\newblock Fuzzy sets.
\newblock {\em Information and Control}, 8(3):338--353, 1965.

\end{thebibliography}
\end{document}